\documentclass[11pt,letter]{article} \usepackage{amsthm}
\usepackage{fullpage}

\usepackage{amsmath, amssymb, graphicx, multirow, epsfig, subfigure}
\usepackage{bbold, array}
\usepackage{enumitem}
\usepackage{natbib,url}%\usepackage{natbib,hypernat,url}
\usepackage[usenames,dvipsnames]{xcolor}
\usepackage[colorlinks,citecolor=BlueViolet,urlcolor=blue,linkcolor=RedViolet]{hyperref}

\usepackage{algpseudocode,algorithm,algorithmicx}

\usepackage{amsmath,amssymb}

\def\argmax{\operatornamewithlimits{arg\,max}}

\def\a{\alpha}
\def\b{\beta}

\newtheorem{definition}{Definition}
\newtheorem{theorem}{Theorem}
\newtheorem{lemma}{Lemma}
\newtheorem{proposition}{Proposition}
\newtheorem{corollary}{Corollary}

\newtheorem{observation}{Observation}
\newenvironment{lemm}[1][]{\noindent {\bf Lemma #1.$\;$}\itshape}{}
\newenvironment{theore}[1][]{\noindent {\bf Theorem #1.$\;$}\itshape}{}
\newenvironment{propositio}[1][]{\noindent {\bf Proposition #1.$\;$}\itshape}{}
\newenvironment{corollar}[1][]{\noindent {\bf Corollary #1.$\;$}\itshape}{}

\newenvironment{sketch}{\noindent {\em Proof sketch.}}{$\qed$\medskip }

\title{GSP with General Independent Click-Through-Rates\footnote{This is a longer version of a conference paper in WINE 2014.}}
\author{
Ruggiero Cavallo\footnote{Yahoo Labs, New York, NY, \url{cavallo@yahoo-inc.com}.} \and Christopher A. Wilkens\footnote{Yahoo Labs, Sunnyvale, CA \url{cwilkens@yahoo-inc.com}.}
}

\begin{document}
\maketitle
\sloppy

\def\arraystretch{1.3}

%%%%%%%%%%%%%%%%%%%%%%%%%%%%%%%%%%%%%%%%%%%%%%%%%%%%%%%%%%

\begin{abstract}
The popular generalized second price (GSP) auction for sponsored search is
built upon a {\em separable} model of click-through-rates that decomposes the
likelihood of a click into the product of a ``slot effect'' and an ``advertiser
effect''---if the first slot is twice as good as the second for some bidder,
then it is twice as good for everyone. Though appealing in its simplicity, this
model is quite suspect in practice. A wide variety of factors including
externalities and budgets have been studied that can and do cause it to be
violated.
In this paper we adopt a view of GSP as an iterated second price auction (see,
e.g., \cite{milgrom10}) and study how the most basic violation of
separability---position dependent, arbitrary public click-through-rates that do
not decompose---affects results from the foundational analysis of
GSP~\citep{varian07,eos07}. For the two-slot setting we prove that for
arbitrary click-through-rates, for arbitrary bidder values, an efficient
pure-strategy equilibrium always exists; however, without separability there
always exist values such that the VCG outcome and payments {\em cannot} be
realized by any bids, in equilibrium or otherwise. The separability assumption
is therefore necessary in the two-slot case to match the payments of VCG but
not for efficiency. We moreover show that without separability, generic
existence of efficient equilibria is sensitive to the choice of tie-breaking
rule, and when there are more than two slots, no (bid-independent) tie-breaking
rule yields the positive result. In light of this we suggest alternative
mechanisms that trade the simplicity of GSP for better equilibrium properties
when there are three or more slots.
\end{abstract}

%==================================================================

\section{Introduction} \label{sec:intro}

The generalized second price (GSP) auction is the predominant auction for
sponsored search advertising today. The auction takes per-click bids and
proceeds as follows: a score is computed {\em independently} for each
advertiser, reflecting its bid and propensity to be clicked; ads are ranked
according to these scores, matched to slots accordingly, and finally charged
the minimum bid required to maintain their allocated slot (i.e., to stay above
the winner of the slot below). Fundamental to this procedure is the fact that the optimal
assignment can be computed based on a ranking of independently-computed scores,
which requires that: (a) the differences between slots affect all ads equally,
and (b) an ad's propensity to be clicked is unaffected by the other ads shown
around it.  Formally, this amounts to {\em separability} of
click-through-rates: any given ad $i$'s probability of being clicked when shown
in slot $j$ decomposes into two factors, $\mu_j$ (a ``slot-effect'') and $\b_i$
(an ``advertiser effect'').  The GSP auction, as well as the theory underlying
it (\cite{varian07,eos07}), all critically rely on this model.

Unfortunately, separability generally does not hold in practice \if For
instance, a number of works have explored how surrounding ads may positively or
negatively affect the likelihood that a user clicks (\fi (one recent work
challenging the model is \cite{jeziorski14}). Moreover, the inadequacy of the
model is becoming more acute as online advertising evolves to incorporate more
heterogeneous bidders and slots.  Instead of a uniform column of vanilla text
ads, it is now common for different ad formats (images, text with sitelinks,
etc.) to appear together on the same search results page. New ad marketplaces
with richer formats, such as Yahoo's ``native'' stream, have emerged. For
advertisers that are seeking clicks, {\em click-through-rate} is the relevant
metric, but for brand advertisers the ``view rate'' of a slot is more relevant
(see, e.g.,~\cite{hummel14}).
\if We will provide evidence that these two quantities do not decay at nearly
the same rate as one goes down the list of slots; and so separability is an
especially outlandish assumption in this context. \fi

%The failure of the separability assumption is widely acknowledged, but the
%model has nonetheless been useful as as a way of simplifying analysis and
%facilitating progress on a host of interesting design questions.
%
In this paper we examine what happens if we move beyond the separable model:
besides assuming---as in the standard model---that click-through-rates can be
determined independent of context (i.e., surrounding ads), we make virtually no
structural assumptions and determine to what extent the most important
classical findings hold up.

Our main results in the two-slot setting show that efficiency is achievable but
revenue may suffer. For arbitrary click-through-rates and values, there exist
efficient equilibria.  However, for arbitrary click-through-rates, there exist
values such that the VCG outcome and payments are {\em not} achievable (in
equilibrium or otherwise). Put another way:
{\em all} click-through-rate profiles ensure existence of efficient equilibria,
but {\em no} non-separable click-through-rate profiles ensure the feasibility
of VCG payments.
We also show that the price of anarchy in a two-slot setting without the
separability assumption is 2 (\cite{caragiannis14} showed that it is at most
1.282 with separability).

When there are three or more slots, we show that efficient equilibria do not
always exist if the tie-breaking rule cannot be chosen dynamically in response
to bids. We present an alternate mechanism that restores efficient equilibria
by expanding the bid space so that agents can specify a bid for every slot,
with items left unallocated if there is not sufficient competition.

\subsection{Related work}

Besides providing one of the earliest models of the sponsored search setting,
\cite{eos07} proved that---in the complete-information model with separable
click-through-rates---GSP has an equilibrium that realizes the VCG result,
i.e., an efficient allocation with each winner paying a price equal to the
negative externality his presence exerts on the other advertisers. In another
important early paper, \cite{lahaie06} provides equilibrium analysis for GSP
(including for the version where advertiser effects are ignored) and
first-price variants, in both the complete and incomplete information settings.
A good early survey is \cite{lahaie07}.

In a recent paper, \cite{caragiannis14} examine the space of equilibria that
may exist under GSP with separable click-through-rates, and bound the
efficiency loss that can result in any of the sub-optimal equilibria. Part of
this work involves a straightforward price of anarchy analysis for the complete
information setting, \if with separable click-through-rates, \fi to which we
provide a counterpoint without separability in Section \ref{sec:poa}.

The prior literature contains some empirical evidence against the separability
assumption. For instance, \cite{craswell08} demonstrate clear violations of
separability for {\em organic} search results. \cite{gomes09} take three
prominent keywords and show that the separable model is a poorer fit to
observed clicks than an alternate ``ordered search'' model of
click-through-rates.
Most of the work that steps outside of the classic separable model is motivated
by externalities between advertisements
\citep{kempe08,ghosh08,giotis08,athey11,aggarwal08,gomes09,ghosh10}. The
context in which an ad is shown may matter: for instance, an ad may yield more
clicks if shown below poor ads than it would if shown below very compelling
competitors. Our model in the current paper removes the separability assumption
but does not capture externalities, as it assumes click-through-rates are
context-independent.

\cite{aggarwal06} show that in the absence of the separability assumption,
there are cases where truthful bidding under GSP will not lead to an efficient
allocation. The authors go on to design a truthful mechanism that implements
the allocation that would result under GSP (which is not truthful) with
truthful bidding. \cite{gonen08} extend this analysis in a setting with reserve
prices. 

Finally, without separability a set of agents could have arbitrary ``expected
values'' for each slot---no common structure is assumed. Though types in our
model are single-dimensional since click-through-rates are not private
knowledge, there is a connection to work that shows existence of efficient
equilibria when agents have a private value for each slot
\citep{leonard83,abrams07}.

%\vspace{2mm} \noindent {\bf MORE::} \cite{bachrach14}

%==================================================================

\section{Preliminaries} \label{sec:gsp}

The basic sponsored search model can be described as follows: a set of $m$
slots are to be allocated among \if at most \fi $n\geq m$ advertisers. When ad
$i$ is shown in slot $j$, regardless of what is shown in other slots, a user
clicks on ad $i$ with probability (``click-through-rate'') $\a_{i,j}$,
generating value $v_i$ for the advertiser. We let $I$ denote the set of
advertisers, and assume throughout that lower slots yield weakly lower
click-through-rates, i.e., $\forall i \in I,\, \forall k \in
\{1,\ldots,m-1\},\, \a_{i,k} \geq \a_{i,k+1}$, and that $\forall i \in I,\,
\forall k \in \{1,\ldots,m\},\, \a_{i,k} > 0$. Our model places no further
assumptions on click-through-rates.

In the {\em separable} refinement of this model, click-through-rates
$\alpha_{i,j}$ can be decomposed multiplicatively into
$\alpha_{i,j}=\mu_j\beta_i$, where $\mu_j$ is the {\em slot effect} that
depends only on the position and $\beta_i$ is the {\em ad effect} that depends
only on the bidder. Slots are ordered so that
$\mu_1\geq\mu_2\geq\dots\geq\mu_m$.
In that  setting, the GSP auction can be defined like such:
\begin{definition}[GSP auction]
The {\em generalized second price (GSP)} auction proceeds as follows:
\begin{enumerate}[topsep=6pt,itemsep=-0.0ex,partopsep=1ex,parsep=1ex]
\item Each bidder $i \in I$ submits a per-click bid $b_i$.
\item Bidders are ranked by $\beta_i b_i$ and matched to slots according to their rank.
\item The bidder in position $j$ pays ``the ad-effect-adjusted bid of the bidder in position $j+1$'' when her ad is clicked; specifically, she pays the minimum amount required to be ranked in position $j$:
\[p_j = \frac{\beta_{j+1}b_{j+1}}{\beta_j}\]%\enspace.\]
\end{enumerate}
\end{definition}

To move beyond separable click-through-rates, we must generalize the GSP mechanism. We will work from a common observation (see, e.g., ~\cite{milgrom10}) that the GSP auction can be viewed as a special sequence of second-price auctions---each slot is sold in order as if it were the only slot for sale. This view allows us to naturally handle general click-through-rates.
\begin{definition}[Iterated second price auction]  \label{ispa}
An {\em iterated second price auction} for sponsored search proceeds as follows:
\begin{enumerate}[topsep=6pt,itemsep=-0.0ex,partopsep=1ex,parsep=1ex]
\item Each bidder $i \in I$ submits a per-click bid $b_i$.
\item An order-of-sale $\sigma$ is selected for the slots.
\item For $j$ from 1 to $m$, with slots indexed according to $\sigma$:
\begin{enumerate}[topsep=1pt,itemsep=-0.0ex,partopsep=1ex,parsep=1ex]
\item A second-price auction is used to sell slot $j$ as follows: let $i^*$ be the remaining bidder with the highest $\alpha_{i,j}b_i$ and let $i^+$ be the bidder with the second-highest $\alpha_{i,j}b_i$. $i^*$ wins the auction and pays $\frac{\alpha_{i^+,j}b_{i^+}}{\alpha_{i^*,j}}$ per click or $\alpha_{i^+,j}b_{i^+}$ per impression.
\item Bidder $i^*$ is removed from the auction and cannot win future slots.
\end{enumerate}
\end{enumerate}
\end{definition}

\noindent This auction, with ``best to worst'' as the order of sale adopted in
step 2, is the implicit context for all results in this paper, except where
stated otherwise.

Another important auction mechanism is the Vickrey-Clarke-Groves (VCG) mechanism, which yields tuthful bidding and an effcient allocation in dominant strategies.
\begin{definition}[VCG mechanism]
In sponsored search, the {\em Vickrey-Clarke-Groves (VCG) mechanism} proceeds as follows:
\begin{enumerate}[topsep=6pt,itemsep=-0.0ex,partopsep=1ex,parsep=1ex]
\item Each bidder $i \in I$ submits a per-click bid $b_i$ to the auction.
\item Bids are interpreted as values per-click, and a matching of bidders to slots $i(j)$ is chosen that maximizes welfare, i.e that maximizes $\sum_{j \in I}\alpha_{i(j),j}b_{i(j)}$.
\item Each bidder $i \in I$ is charged an amount equal to the welfare other bidders would gain, according to their reported bids, if $i$ were removed from the auction. %CW comment: should we define this better?
\end{enumerate}
\end{definition}
In this paper we will compare the outcome of our GSP generalization (Definition \ref{ispa}) to that of VCG:
\begin{definition}[VCG result]
The {\em VCG result} refers to the allocation and payments realized by the VCG mechanism.
\end{definition}

Given a set of advertiser bids, if we say that an auction has ``realized the
VCG result'' we are saying that its allocation and payments match those of the
VCG mechanism.

%==================================================================

\section{Two slots} \label{sec:2item}

We focus much of our analysis on the two-slot case, for a few reasons. First,
this is the simplest case in which GSP deviates from a straightforward Vickrey
auction (which it reduces to in the case of a single slot); second, with two
slots ``separability'' is cleanly and simply defined, holding whenever the
ratio of click-through-rate for the first slot to the click-through-rate for
the second (henceforth, the {\em click-ratio}) is the same across all agents
(i.e., $\forall i, j \in I,\,
\frac{\a_{i,1}}{\a_{i,2}}=\frac{\a_{j,1}}{\a_{j,2}}$);
and finally, we will be able to show important positive results for the
two-slot case that do not extend to larger numbers of slots.

%third, we have only been able to prove our strongest results for the two-slot
%case.

%-------------------------------------------------------

\subsection{Efficient equilibria}  \label{sec:eff}

Among the first questions one might ask about an auction mechanism is: does it
yield {\em efficient} equilibria? The foundational work of \cite{eos07} and
\cite{varian07} demonstrated that efficient equilibria do exist under GSP in
the separable model, and we now ask whether the assumption of separability is
necessary. We resolve this in the negative.

To build intuition, we will start by considering an especially ``problematic''
example for the non-separable setting that reveals some of the challenges
that can arise.

\begin{table}[h!]
\begin{center}
\begin{tabular}{c|c|c|c}
bidder& value& $\a_{i,1}$& $\a_{i,2}$ \\
\hline 
$1$& $1$& $1$& $1$ \\
$2$& $1$& $1$& $1$ \\
$3$& $2$& $0.4$& $0.2$
\end{tabular}
\vspace{0mm}
\caption{\label{tab:e0} A two-slot, three-bidder example in which two bidders are indifferent between the two slots. There is no pure strategy equilibrium unless ties are broken in favor of bidder 3.}
\vspace{-0mm}
\end{center}
\end{table}

%Consider the valuations and click-through-rates depicted in Table
%\ref{tab:e0}.
The first thing to notice about the example in Table~\ref{tab:e0} is that in
any pure strategy equilibrium bidders 1 and 2 win the slots and bidder 3 gets
nothing.\footnote{Assume otherwise. If bidder 3 were winning the first slot in
equilibrium, he must be paying less than his value in expectation (0.8), but in
that case the loser amongst bidders 1 and 2 could benefit by bidding between
0.8 and 1, winning a slot for at most 0.8. Likewise, if bidder 3 were winning
the second slot in equilibrium, he must be paying less than his value in
expectation (0.4), but in that case the loser amongst bidders 1 and 2 could
benefit by bidding between 0.4 and 1.} So now assume without loss of generality
that bidder 1 wins slot 1 and bidder 2 wins slot 2. Bidder 2 has the better
deal, since he'll have to pay at most half of what bidder 1 pays (since the bid
of bidder 3 will set the price for slot 2 and lower-bound the price for slot
1), and slot 2 is as good as slot 1 (in the eyes of bidders 1 and 2). Thus, in
order for bidder 1 to be best-responding, it must be impossible for him to bid
so as to win slot 2 (which would sell at a lower price) instead of slot 1. In
other words, if he were to underbid bidder 2, he must end up with nothing. This
can only be the case if $0.4 b_3 = b_2$ {\em and} a hypothetical tie between
bidders 2 and 3 for slot 1 is broken in favor of bidder 3.

Therefore, interestingly, in the above example there exists a pure strategy
equilibrium---efficient or otherwise---{\em only} if ties are broken in a
specific way. This would seem to bode very poorly for the prospects of a
general result establishing existence of efficient equilibria. However, we will
see in this section that, at least in the two-slot case, efficient equilibria
do in fact always exist (given the right tie-breaking rule).
% Therefore the
%separability assumption of \cite{varian07} and \cite{eos07} turns out to be
%unnecessary for efficiency.
%
Several of the proofs are somewhat painstaking and are deferred to the
Appendix, along with auxiliary lemmas, with proof sketches included in the main
text.

\begin{theorem}[Efficient equilibria exist]  \label{the:eff}
In a two-slot setting with any number of bidders, for arbitrary values and
click-through-rates, if there is a unique efficient allocation and ties are
broken in favor of an agent with highest click-ratio, then there is an
efficient equilibrium without overbidding.
\end{theorem}

\begin{sketch}
Consider arbitrary click-through-rates $\a$ and values $v$. Let 1 and 2 denote
the respective winners of slots 1 and 2 in the efficient allocation, and let 3
denote $\argmax_{j \in I\setminus\{1,2\}} \a_{j,2}v_j$. We dichotomize the set
of possible click-through-rate profiles into those in which
$\frac{\a_{2,1}}{\a_{2,2}} \geq \frac{\a_{3,1}}{\a_{3,2}}$ and those in which
$\frac{\a_{2,1}}{\a_{2,2}} < \frac{\a_{3,1}}{\a_{3,2}}$. In the former case,
the following bid profile is an efficient equilibrium: $b_1=v_1$,
$b_2=\frac{\a_{3,2}v_3+(\a_{2,1}-\a_{2,2})v_2}{\a_{2,1}}$, $b_3=v_3$, and
$b_i=0$, $\forall i \in I\setminus\{1,2,3\}$. In the latter, the following is:
$b_1=v_1$, $b_2=\frac{\a_{3,2}}{\a_{2,2}} v_3$,
$b_3=\frac{\a_{2,1}}{\a_{2,2}}\frac{\a_{3,2}}{\a_{3,1}}v_3$, and $b_i=0$,
$\forall i \in I\setminus\{1,2,3\}$. The proof verifies that an exhaustive set
of sufficient equilibrium conditions holds in each case.
\end{sketch}

It is interesting to see what the above tells us about the problematic example
of Table \ref{tab:e0}. To derive bids yielding an efficient equilibrium, we can
note the following about the example: an efficient allocation gives slots 1 and
2 to bidders 1 and 2, and thus the agents are labeled in a way consistent with
the convention of Theorem \ref{the:eff}. Now, since $\frac{\a_{2,1}}{\a_{2,2}}
< \frac{\a_{3,1}}{\a_{3,2}}$, the above proof indicates that the following
bids---combined with a tie-breaking rule that favors bidder 3 over bidder
2---yields an efficient equilibrium: $b_1 = v_1 = 1$, $b_2 =
\frac{\a_{3,2}}{\a_{2,2}} v_3 = 0.4$, $b_3 =
\frac{\a_{2,1}}{\a_{2,2}}\frac{\a_{3,2}}{\a_{3,1}}v_3 = 1$.

%-------------------------------------------------------

\subsection{Globally envy-free equilibria}  \label{sec:gef}

In the previous subsection we demonstrated that efficient equilibria always
exist. We proved this constructively, giving bid functions that yield efficiency
for all valuations. However, these bids do not generally lead to global
envy-freeness:

\begin{definition}[Globally envy-free outcome]
Consider an arbitrary allocation and prices. Let $k$ denote the winner of slot
$k$ and $p_k$ denote the price paid by $k$, for $k \in \{1,\ldots,n\}$; let
$n+1$ denote the agent that receives nothing; and let $p_{n+1}=0$. The
allocation and prices constitute a globally envy-free outcome if and only if,
for all $i,j \in \{1,\ldots,n+1\}$,
\[ \a_{i,i} v_i - p_i \geq \a_{i,j} v_i - p_j  \]
\end{definition}

Envy-freeness is a major focus of the classic work on GSP
\citep{eos07,varian07}, because of its relationship to VCG results, the
salience it arguably confers on equilibria, and perhaps most importantly, the
fact that envy-freeness implies that an equilibrium generates at least as much revenue as VCG.
Unfortunately, we will now see that this guarantee does not extend to our
setting, and an envy-free equilibrium is not guaranteed to exist.

We will give a necessary condition for global envy-freeness (Proposition
\ref{prop:gef0}), which will not always be satisfied. We will then show that
whenever the condition is satisfied, global envy-freeness {\em can} be
achieved, and moreover done so in the context of an efficient equilibrium
(Theorem \ref{the:gef}).

\begin{proposition}  \label{prop:gef0}
In a two-slot, three-bidder setting, for arbitrary values $v$ and click-through-rates
$\a$, there exist no bids yielding a globally envy-free outcome unless, letting
1 and 2 denote the respective winners of slots 1 and 2 in the efficient
allocation:
\begin{align*} &(\a_{3,1}-\a_{3,2})v_3 \leq (\a_{1,1}-\a_{1,2}) v_1
\end{align*}
\end{proposition}

\begin{proof}
Take arbitrary values $v$, click-through-rates $\a$, and bids $b$.
First assume $b_2 \geq \frac{\a_{3,1}}{\a_{2,1}}b_3$ (i.e., 2 sets the price
for 1). For 1 to not be envious of 2, it must be the case that $b_2 \leq
\frac{\a_{3,2}b_3 + (\a_{1,1}-\a_{1,2})v_1}{\a_{2,1}}$. The combination of
these two constraints yields $(\a_{3,1}-\a_{3,2})b_3 \leq (\a_{1,1}-\a_{1,2})
v_1$. Now instead assume $b_2 \leq \frac{\a_{3,1}}{\a_{2,1}}b_3$ (i.e, 3 sets
the price for 1). 1 is not envious of 2 if and only if $\a_{1,1}v_1-\a_{3,1}b_3
\geq \a_{1,2}v_1-\a_{3,2}b_3$, i.e., $(\a_{3,1}-\a_{3,2})b_3 \leq
(\a_{1,1}-\a_{1,2}) v_1$, again. Finally, noting that envy-freeness for 3
requires that $b_3 \geq v_3$ (otherwise 3 would envy 2), global envy-freeness
requires:
\[ (\a_{3,1}-\a_{3,2})v_3 \leq (\a_{3,1}-\a_{3,2})b_3 \leq (\a_{1,1}-\a_{1,2})
v_1 \]
If $(\a_{3,1}-\a_{3,2})v_3 > (\a_{1,1}-\a_{1,2}) v_1$, this cannot be
satisfied.
\end{proof}

\begin{theorem}[GEF and efficient equilibria condition]  \label{the:gef}
In a two-slot, three-bidder setting, for arbitrary click-through-rates $\a$ and values
$v$, there exist bids yielding a globally envy-free and efficient equilibrium
if and only if, letting 1 and 2 denote the respective winners of slots 1 and 2
in the efficient allocation:
\begin{align*} &(\a_{3,1}-\a_{3,2})v_3 \leq (\a_{1,1}-\a_{1,2}) v_1
\end{align*}
If a globally envy-free and efficient equilibrium exists, one exists that
yields the VCG result and does not require overbidding.
\end{theorem}

\begin{sketch}
The full proof follows similar lines to that of Theorem \ref{the:eff}, and is
again relegated to the Appendix along with auxiliary lemmas. The proof
considers three cases: (i) $(\a_{3,1}-\a_{3,2})v_3 \leq
(\a_{2,1}-\a_{2,2})v_2$; (ii) $(\a_{3,1}-\a_{3,2})v_3 > (\a_{2,1}-\a_{2,2})v_2$
and $\frac{\a_{3,1}}{\a_{2,1}}v_3 \leq v_2$; and (iii) $(\a_{3,1}-\a_{3,2})v_3
> (\a_{2,1}-\a_{2,2})v_2$ and $\frac{\a_{3,1}}{\a_{2,1}}v_3 > v_2$. In each
case it is assumed that $(\a_{3,1}-\a_{3,2})v_3 \leq (\a_{1,1}-\a_{1,2})v_1$.
We specify bids yielding efficient and globally envy-free equilibria: in case
(i) $b_1=v_1$, $b_2=\frac{\a_{3,2}v_3+(\a_{2,1}-\a_{2,2})v_2}{\a_{2,1}}$, and
$b_3=v_3$; in case (ii), $b_1=v_1$, $b_2=\frac{\a_{3,1}}{\a_{2,1}}v_3$, and
$b_3=v_3$; and in case (iii), $b_1=v_1$, $b_2=v_2$, and $b_3=v_3$.
\end{sketch}

For instance, consider a 3-agent example with $v_1=v_2=v_3=1$, $\a_{1,1}=0.9$,
$\a_{1,2}=0.5$, $\a_{2,1}=0.5$, $\a_{2,2}=0.4$, $\a_{3,1}=0.6$, and
$\a_{3,2}=0.1$. The unique efficient allocation gives slots 1 and 2 to bidders
1 and 2, respectively. But we have: $0.5 = (\a_{3,1}-\a_{3,2}) v_3 >
(\a_{1,1}-\a_{1,2}) v_1 = 0.4$. Thus Theorem \ref{the:gef} implies that there
can be no globally envy-free and efficient equilibrium.

A characterization for more than three bidders is harder to state in a concise
form, but the following theorem gives sufficient conditions for efficiency and
global envy-freeness.

\begin{theorem}  \label{the:gef2}
For arbitrary click-through-rates and values, letting 1 and 2 denote the respective
winners of slots 1 and 2 in the efficient allocation, if
$\frac{\a_{2,1}}{\a_{2,2}} \geq \frac{\a_{i,1}}{\a_{i,2}}$, $\forall i \in
I\setminus\{1,2\}$, there exists an efficient and globally envy-free
equilibrium without overbidding.
\end{theorem}

%-------------------------------------------------------

\subsection{VCG results cannot always be achieved}  \label{sec:vcg}

In the results of \cite{eos07}, existence of an efficient equilibrium in the
separable setting is demonstrated via proof that an equilibrium realizing the
VCG result always exists. In some sense the VCG result is the most salient
kind of efficient equilibrium, and it would be surprising if efficient
equilibria exist generically but VCG equilibria do not. But that is exactly
what we now demonstrate. Whenever a set of click-through-rates violates
separability,\footnote{Interestingly, there is one very minor exception to this
``whenever'': if there is exactly one agent whose click-ratio is not equal to
the maximum across all bidders---i.e., click-through-rates are separable except in the
case of one bidder, and his click-ratio is {\em lower}---then the VCG result
will be supported.} one can never be assured that a VCG result is feasible, in
equilibrium or otherwise. That is, there always exist values that make it
impossible for the agents to bid in a way that yields an efficient allocation
and VCG prices.

\begin{theorem}[Always a bad value profile]  \label{the:vcg}
Assume strictly decreasing click-through-rates. In a two-slot setting with three
bidders, one of whom has a strictly higher click-ratio than the other two,
there always exist values such that the VCG result is not supported.
\end{theorem}

\begin{proof}
Consider three agents with strictly decreasing click-through-rates $\a$ such
that one agent's click-ratio is strictly higher than that of the other two.
Label the three bidders in a non-decreasing order of $\a_{i,1}/\a_{i,2}$.
Strictly decreasing click-through-rates entails that $\a_{1,1}/\a_{1,2} > 1$,
and $\a_{1,1}/\a_{1,2} \leq \a_{2,1}/\a_{2,2} < \a_{3,1}/\a_{3,2}$ by
assumption. In other words, for some $\epsilon \geq 0$ and $\delta > 0$,
\begin{align}  1 < \frac{\a_{1,1}}{\a_{1,2}} = \frac{\a_{2,1}}{\a_{2,2}} -
\epsilon = \frac{\a_{3,1}}{\a_{3,2}} - \epsilon - \delta   \label{t0}
\end{align}

Fix arbitrary $v_3 > 0$. Let $\lambda_1 = \big(
\frac{\a_{3,1}-\a_{3,2}}{\a_{1,1}-\a_{1,2}} - \frac{\a_{3,1}}{\a_{1,1}} \big)
v_3$. Note that:
\begin{align*}
\frac{\a_{3,1}-\a_{3,2}}{\a_{1,1}-\a_{1,2}} > \frac{\a_{3,1}}{\a_{1,1}}
 \;\;\Leftrightarrow\;\; &1 - \frac{\a_{3,2}}{\a_{3,1}} > 1 -
\frac{\a_{1,2}}{\a_{1,1}}
 \\ \Leftrightarrow\;\; &\frac{\a_{1,1}}{\a_{1,2}} < \frac{\a_{3,1}}{\a_{3,2}}
\end{align*}
This holds by (\ref{t0}), and thus $\lambda_1 > 0$. Now let $\lambda_2 = \big(
\frac{\a_{3,1}-\a_{3,2}}{\a_{2,1}-\a_{2,2}} - \frac{\a_{3,2}}{\a_{2,2}} \big)
v_3$. Note that:

\begin{align*}
\frac{\a_{3,1}-\a_{3,2}}{\a_{2,1}-\a_{2,2}} > \frac{\a_{3,2}}{\a_{2,2}}
 \;\;\Leftrightarrow\;\; &\frac{\a_{3,1}}{\a_{3,2}} - 1 >
\frac{\a_{2,1}}{\a_{2,2}}
- 1
\\ \Leftrightarrow\;\; &\frac{\a_{2,1}}{\a_{2,2}} < \frac{\a_{3,1}}{\a_{3,2}}
\end{align*}
This also holds by (\ref{t0}), and thus $\lambda_2 > 0$. Now let $v_1 =
\frac{\a_{3,1}-\a_{3,2}}{\a_{1,1}-\a_{1,2}} v_3 - \gamma_1$, for some $\gamma_1
\in (0,\lambda_1)$. We have:
\begin{align} &\frac{\a_{3,1}}{\a_{1,1}} < \frac{v_1}{v_3} <
\frac{\a_{3,1}-\a_{3,2}}{\a_{1,1}-\a_{1,2}} \label{t1} \end{align}

\noindent And let $v_2 = \frac{\a_{3,2}}{\a_{2,2}}v_3 + \gamma_2$, for some
$\gamma_2 \in (0,\lambda_2)$. We have:
\begin{align} &\frac{\a_{3,2}}{\a_{2,2}} < \frac{v_2}{v_3} <
\frac{\a_{3,1}-\a_{3,2}}{\a_{2,1}-\a_{2,2}} \label{t2} \end{align}

\noindent We refine our specification of $\gamma_1$ and $\gamma_2$ such that:
\begin{align*} &\frac{1}{\a_{3,2}v_3}\big[ (\a_{1,1}-\a_{1,2})\gamma_1 +
(\a_{2,1}-\a_{2,2}) \gamma_2 \big] < \delta,  %\label{eVCG1}
\\ &\gamma_2 > \frac{\a_{1,1}-\a_{1,2}}{\a_{2,2}} \gamma_1 \text{, and}
%\label{eVCG2}
%
\\ &\frac{\a_{1,1}-\a_{1,2}}{\a_{1,1}v_3} \Big( \frac{\a_{1,1}}{\a_{3,2}}
\gamma_1 + \frac{\a_{2,1}-\a_{2,2}}{\a_{3,2}} \gamma_2 \Big) < \delta +
\frac{\a_{1,2}}{\a_{1,1}}\epsilon  %\label{eVCG3}
\end{align*}
Note that such values can be chosen consistent with everything specified above,
for arbitrary $\delta > 0$.

Letting $(i,j)$ denote the allocation in which agent $i$ receives slot 1 and
agent $j$ receives slot 2,
%: $(1,2)$ is an efficient allocation. 
\if , $(3,2)$ is an efficient allocation in the absence of agent 1, and $(3,1)$
is an efficient allocation in the absence of agent 2.  \fi if we can establish
that (1,2) is an efficient allocation, then the bidder labels here correspond
to those used in Proposition \ref{prop:gef0}.
Letting $w(i,j)$ denote $\a_{i,1}v_i + \a_{j,2} v_j$, i.e., the social value of
allocation $(i,j)$, this can be established by demonstrating that: $w(1,2) >
w(2,1)$, $w(1,2) > w(1,3)$, $w(1,2) > w(2,3)$, $w(1,2) > w(3,2)$, and $w(1,2) >
w(3,1)$.
Due to space constraints, we omit demonstration of these inequalities, which
are relatively straightforward.

\if 0  %%%%%%%%%%%%%%%%%%
We can now observe that $p_1 = \a_{3,1}v_3$ and $p_2 = \a_{3,1}v_3 -
(\a_{1,1}-\a_{1,2})v_1$. Therefore:
%
\begin{align*}
%
\frac{p_1}{p_2} = &\frac{\a_{3,1}v_3}{\a_{3,1}v_3 - (\a_{1,1}-\a_{1,2})v_1}
%
\\ =\; &\frac{\a_{3,1}v_3}{\a_{3,1}v_3 - (\a_{1,1}-\a_{1,2}) \Big(
\frac{\a_{3,1}-\a_{3,2}}{\a_{1,1}-\a_{1,2}} v_3 - \gamma_1 \Big)}
%
\\ =\; &\frac{\a_{3,1}v_3}{\a_{3,2}v_3 + (\a_{1,1}-\a_{1,2}) \gamma_1 }
%
\\ <\; &\frac{\a_{3,1}}{\a_{3,2}}
%
\end{align*}
%
The inequality holds because $\a_{1,1}-\a_{1,2} > 0$, by the strictly
decreasing click-through-rate assumption, and $\gamma_1 > 0$. In light of Corollary
\ref{cor0}, this completes the proof.
\fi   %%%%%%%%%%%%%%%%%

Now, since a VCG result is always globally envy-free (see, e.g.,
\cite{leonard83}), in light of Proposition \ref{prop:gef0}, to complete the
proof it is sufficient to show that $(\a_{3,1}-\a_{3,2})v_3 >
(\a_{1,1}-\a_{1,2})v_1$. We have:
\begin{align*}
&(\a_{3,1}-\a_{3,2})v_3 - (\a_{1,1}-\a_{1,2})v_1
\\ =\; &(\a_{3,1}-\a_{3,2})v_3 - (\a_{1,1}-\a_{1,2})\bigg(\frac{\a_{3,1}-\a_{3,2}}{\a_{1,1}-\a_{1,2}}v_3 - \gamma_1 \bigg)
\\ =\; &(\a_{1,1}-\a_{1,2})\gamma_1 \\ >\; &0
\end{align*}
\end{proof}

The result extends almost immediately to the $n$-bidder case if we forbid
overbidding (note that overbidding is weakly dominated in GSP). We can fix the
values of all but three agents to 0; then the problem is equivalent to one in
which the 0-valued agents do not exist, since they can't bid anything other
than 0.

\begin{corollary}  \label{cor1}
Assume strictly decreasing click-through-rates. In a two-slot setting with any
number of bidders greater than two, if there exists a bidder with click-ratio
strictly greater than that of two other agents, there always exist values such
that the VCG result is not supported without overbidding.
\end{corollary}

In light of this negative result, one might ask whether the VCG result can be
recovered if we are willing to experiment with different orders of sale. It
turns out this can never help in the two-slot case.

\begin{proposition}  \label{prop:inorder}
In settings with at most three bidders, if the VCG result is not supported
when selling slots in-order, it is not supported when selling slots in reverse
order.
\end{proposition}

\begin{proof}
Let 1 and 2 denote the bidders that receive items 1 and 2, respectively, in a
VCG result, and let $p_1$ and $p_2$ denote the respective (per-impression) VCG
prices.
Suppose first that we sell the items in order to achieve the VCG result. Since bidder 3 will be the only competition for bidder 2, it must be that $\a_{3,2}b_2=p_2$. Moreover, we can suppose that $\a_{2,1}b_2=p_1$ (lowering $b_2$ cannot help, bidding higher will interfere with the auction for item 1 either by winning the item or by raising the price). Thus, suppose bidders bid as follows:
\[b_1=v_1\enspace,\quad b_2 = \frac{p_1}{\a_{2,1}}\enspace,\quad\mbox{and}\quad b_3=\frac{p_2}{\a_{3,2}}\enspace.\]
 
By construction, these bids will achieve the VCG result as long as two other conditions are met: $\alpha_{2,2}b_2\geq p_2$ so bidder 2 still wins item 2, and $\alpha_{3,1}b_3\leq p_1$ so bidder 3 does not interfere in the sale of item 1. The first condition is always true --- envy-freeness of VCG prices implies $\alpha_{2,1}v_2-p_1\leq\alpha_{2,2}v_2-p_2$ and so
\[\alpha_{2,2}(v_2-b_2)\leq\alpha_{2,1}(v_2-b_2)=\alpha_{2,1}v_2-p_1\leq\alpha_{2,2}v_2-p_2\]
\[\alpha_{2,1}b_2\geq p_2\]
as desired. The second condition may indeed be violated.
 
It remains to show that whenever $\a_{3,1}b_3>p_1$, then selling items in reverse order cannot achieve the VCG result. Suppose we find bids that support the VCG result selling out of order. Then bidder 1 must choose a bid $b_1$ that wins item 1 without interfering in the auction for item 2, i.e., a bid $b_1$ such that $\a_{1,2}b_1\leq p_2$ and $\a_{1,1}b_1\geq p_1$. We thus get
\[\a_{1,1}v_1-p_1\geq\a_{1,1}(v_1-b_1)>\a_{1,2}(v_1-b_1)\geq\a_{1,2}v_1-p_2\]
\[\a_{1,1}v_1-p_1>\a_{1,2}v_1-p_2\]
Now, since VCG prices are the minimal envy-free prices (see \cite{leonard83}), some bidder's envy constraint must be tight for item 2 (otherwise we could lower the price of item 2 while preserving envy-freeness). It cannot be bidder 3 because, when $\a_{3,1}b_3>p_1$, bidder 3 strictly prefers item 1 at VCG prices:
\[\a_{3,2}v_3-p_2=\a_{3,2}(v_3-b_3)\leq\a_{3,1}(v_3-b_3)<\a_{3,1}v_3-p_1\enspace.\]
The only remaining bidder who can be indifferent is 1, so we can conclude that $\a_{1,1}v_1-p_1=\a_{1,2}v_1-p_2$, which contradicts the prior statement that $\a_{1,1}v_1-p_1>\a_{1,2}v_1-p_2$. Thus, when $\a_{3,1}b_3>p_1$, selling items in reverse order cannot support the VCG result either.
\end{proof}

%-------------------------------------------------------

\subsection{Price of anarchy}  \label{sec:poa}

%GSP will typically have many pure Nash equilibria.
We established in Section \ref{sec:eff} that our generalization of GSP will
always have an efficient equilibrium, but there may be many inefficient
equilibria as well. In this section we consider how much efficiency may be lost
if one of those other equilibria occurs. We will make the natural assumption
that agents don't bid more than their value; this is standard in the
literature---overbidding is a weakly dominated strategy, and with overbidding
very strange equilibria can be constructed.

We find that the efficient equilibrium is never more than twice as good as the
worst equilibrium, and this bound is tight. This result stands in contrast to
the results of \cite{caragiannis14}, who showed that in the separable setting
with two slots, the efficient equilibrium is never more than 28.2 percent
better than (i.e., yields no more than 1.282 times the social welfare of) the
worst.  One could thus say there is a significant added ``efficiency risk'' in
a setting without separability.

\begin{definition}[Price of anarchy]
Given click-through-rates $\a$ and values $v$, the price of anarchy is the ratio of the
social welfare in the efficient (best) equilibrium to that in the worst
equilibrium; i.e., letting 1 and 2 denote the respective winners of slots 1 and
2 in the efficient allocation, letting $A$ denote the set of equilibrium
allocations, and letting $a_1$ and $a_2$ denote the respective winners of slots
1 and 2 in allocation $a \in A$,
\[ \frac{\a_{1,1}v_1 + \a_{2,2}v_2}{\min_{a \in A} \big[
\a_{a_1,1}v_{a_1}+\a_{a_2,2}v_{a_2} \big]}  \]
\end{definition}

The following lemma, and especially its corollary, will be critical for the
proof bounding price of anarchy in our setting (these proofs are in the
Appendix).

\begin{lemma}  \label{lem:poa}
Let $(i,j)$ denote an allocation in which $i$ receives slot 1 and $j$ receives
slot 2. For arbitrary click-through-rates $\a$ and values $v$, letting 1 and 2 denote
the respective winners of slots 1 and 2 in the efficient allocation, the only
possible inefficient equilibria are: $(\argmax_{i \in I\setminus\{1\}}
\a_{i,1}v_i,\,1)$ and $(2,\,\argmax_{i \in I\setminus\{2\}} \a_{i,2}v_i)$.
\end{lemma}

\begin{proof}
First note that $\forall i \neq 1$, $\a_{2,2}v_2 > \a_{i,2}v_i$, by efficiency.
If 2 is not allocated a slot and slot 2 is allocated to some $j \neq 1$, then
$\a_{2,2}v_2 > \a_{j,2}v_j$, and $b_j \leq v_j$ by assumption, and thus 2 has a
profitable deviation to bid high enough to win slot 2. Thus the only candidates
for equilibria involve 2 receiving a slot or 1 receiving slot 2.

If 1 receives slot 2 in equilibrium, then slot 1 must go to $i = \argmax_{j \in
I\setminus\{1\}} \a_{j,1}v_j$. Otherwise, since bids don't exceed values, $i$
could bid truthfully and win slot 1 for a profit).
If 2 receives slot 1 in equilibrium, then slot 2 must go to $i = \argmax_{j \in
I\setminus\{2\}} \a_{j,2}v_j$. Again, since bids don't exceed values, this
holds because otherwise $i$ could bid truthfully and win slot 2 for a profit.
Finally, if 2 receives slot 2 in equilibrium, then slot 1 must go to $i =
\argmax_{j \in I\setminus\{2\}} \a_{j,1}v_j$. Yet again this holds because
otherwise $i$ could bid truthfully and win slot 1 for a profit. In this case
$i$ is 1, and so $(1,2)$---the efficient allocation---is the only equilibrium
with 2 receiving slot 2.
\end{proof}

\begin{corollary}  \label{corr2}
Given click-through-rates $\a$ and values $v$, letting 1 and 2 denote the respective
winners of slots 1 and 2 in the efficient allocation, letting $j$ denote
$\argmax_{i \in I\setminus\{1\}} \a_{i,1}v_i$ and $k$ denote $\argmax_{i \in
I\setminus\{2\}} \a_{i,2}v_i$, the price of anarchy is:
\[ \max\bigg\{ e(1,2) \cdot 1,\; e(j,1) \cdot \frac{\a_{1,1}v_1 +
\a_{2,2}v_2}{\a_{j,1}v_j+\a_{1,2}v_1},\; e(2,k) \cdot \frac{\a_{1,1}v_1 +
\a_{2,2}v_2}{\a_{2,1}v_2+\a_{k,2}v_k} \bigg\},  \]
where $e(i,j) = 1$ if allocation $(i,j)$ is attainable in
equilibrium\footnote{Note that Theorem \ref{the:eff} entails that $e(1,2)=1$ in
all cases.} and 0 otherwise.
\end{corollary}

\begin{proposition}  \label{prop:poa1}
For the two-slot, $n$-bidder setting, for any $n \geq 2$, for arbitrary
click-through-rates and values, the price of anarchy is at most 2.
\end{proposition}

\begin{proof}
Let 1 and 2 denote the respective winners of slots 1 and 2 in the efficient
allocation, $j$ denote $\argmax_{i \in I\setminus\{1\}} \a_{i,1}v_i$, and $k$
denote $\argmax_{i \in I\setminus\{2\}} \a_{i,2}v_i$. Take arbitrary bids $b$
that realize allocation $(j,1)$ in equilibrium, if any exist.
Let $p_2$ denote the price paid by 1, and let $p_{1,1}$ denote the price 1
would have to pay were he to deviate from the equilibrium in a way that leads
him to win slot 1.
Since $b$ forms an equilibrium, $\a_{1,1}v_1-p_{1,1} \leq \a_{1,2}v_1-p_2$, and
noting that $p_{1,1} \leq \a_{j,1}v_j$, we have:
\[  \a_{1,1}v_1 - \a_{1,2}v_1 \leq p_{1,1} - p_2 \leq \a_{j,1}v_j - p_2  \]

\noindent Adding $\a_{2,2}v_2$ to both sides of this inequality and rearranging
yields:
\[ \a_{1,1}v_1 + \a_{2,2}v_2 \leq \a_{j,1}v_j + \a_{1,2}v_1 + \a_{2,2}v_2 - p_2 \]

\noindent This implies that:
\[ \frac{\a_{1,1}v_1 + \a_{2,2}v_2}{\a_{j,1}v_j+\a_{1,2}v_1}
        \leq \frac{\a_{j,1}v_j + \a_{1,2}v_1 + \a_{2,2}v_2 - p_2}{\a_{j,1}v_j+\a_{1,2}v_1}
        = 1 + \frac{\a_{2,2}v_2 - p_2}{\a_{j,1}v_j+\a_{1,2}v_1}
\]

\noindent Now noting that $\a_{2,2}v_2 \leq \a_{2,1}v_2 \leq \a_{j,1}v_j$ (by
non-decreasing click-through-rates plus the definition of $j$), we have:
\[ 1 + \frac{\a_{2,2}v_2 - p_2}{\a_{j,1}v_j+\a_{1,2}v_1}
        \leq 1 + \frac{\a_{j,1}v_j - p_2}{\a_{j,1}v_j+\a_{1,2}v_1}
        \leq 1 + 1 = 2
\]

Now take arbitrary bids $b$ that realize allocation $(2,k)$ in equilibrium, if
any exist.
Let $p_{1,1}$ denote the price 1 would have to pay were he to deviate from the
equilibrium in a way that leads him to win slot 1, and $p_{1,2}$ the price he'd
have to pay were he to deviate in a way that yields him slot 2.
Since $b$ forms an equilibrium, $\a_{1,2}v_1 \leq p_2 \leq \a_{j,2}v_j$
(otherwise $i$ could bid truthfully and win slot 2 for a profit, and $b_j \leq
v_j$). Similarly, $\a_{1,1}v_1 \leq p_{1,1} \leq \a_{2,1}v_2$ (using $b_2 \leq
v_2$). This implies that:
\begin{align*}
\frac{\a_{1,1}v_1 + \a_{2,2}v_2}{\a_{2,1}v_2+\a_{j,2}v_j}
        \leq\; &\frac{\a_{2,1}v_2 + \a_{2,2}v_2}{\a_{2,1}v_2+\a_{1,2}v_1}
\\	\leq\; &\frac{\a_{2,1}v_2+\a_{2,1}v_2}{\a_{2,1}v_2+\a_{1,2}v_1}
\\	\leq\; &\frac{\a_{2,1}v_2+\a_{2,1}v_2}{\a_{2,1}v_2} = 2
\end{align*}

\noindent We use weakly-decreasing click-through-rates in the second inequality and
non-negativity of values and click-through-rates in the third. By Corollary
\ref{corr2}, this is sufficient to establish the claim.
\end{proof}

We now show that this bound is tight by way of an example.

\begin{proposition}  \label{prop:poa2}
For the two-slot, $n$-bidder setting, for any $n \geq 2$, for arbitrary
$\epsilon > 0$, there exist click-through-rates and values such that the price of
anarchy is at least $2-\epsilon$.
\end{proposition}

\begin{proof}
Consider a setting with $n$ bidders, for arbitrary $n \geq 2$. Consider the
case where two bidders, which we'll call 1 and 2, have value 1 and all other
bidders (if there are any) have value 0. Take $\a_{1,1}=1-\delta$,
$\a_{1,2}=\delta$, $\a_{2,1}=1$, and $\a_{2,2}=1-\delta$, for arbitrary $\delta
\in (0,\frac{1}{3})$. The efficient allocation is $(1,2)$, and this is
supported, e.g., by equilibrium bids $b_1=1$ and $b_2 = \delta$. But allocation
$(2,1)$ is also supported as an equilibrium, e.g., by bids $b_1=0$ and $b_2=1$.
The price of anarchy is thus:
\[ \frac{\a_{1,1}v_1 + \a_{2,2}v_2}{\a_{2,1}v_2+\a_{1,2}v_1} = \frac{(1-\delta)
+ (1-\delta)}{1 + \delta} = \frac{2-2\delta}{1+\delta} \]
For any $\epsilon > 0$,  if $\delta < \frac{\epsilon}{4-\epsilon}$ then
$\frac{2-2\delta}{1+\delta} > 2-\epsilon$. Therefore, for any $\epsilon > 0$,
we can choose $\delta \in (0,\min\{\frac{1}{3},\frac{\epsilon}{4-\epsilon}\})$,
in which case the price of anarchy will exceed $2-\epsilon$.
\end{proof}

This also shows that equilibrium revenue, as a fraction of the VCG revenue, may
be arbitrarily bad. In the example above, the $b_1=0$, $b_2=1$ equilibrium
yields 0 revenue, while the $b_1=1$, $b_2=\delta$ equilibrium yields the VCG
outcome, with revenue $\delta$.

%==================================================================

\section{Three or more slots} \label{sec:3ormore}

So far, we have seen that many of the important properties of the GSP auction break in a two-slot setting. In this section, we will explore additional complexities that arise with more than two slots. Notably, we will see that the order in which slots are sold becomes critical --- it will no-longer be sufficient to sell slots from ``best to worst'' as in a standard GSP auction.

\subsection{Absence of equilibrium}
 
First, we show that even the existence of equilibrium is in doubt. The following example with 4 bidders and 3 slots illustrates that no bid-independent tie-breaking rule can guarantee the existence of an equilibrium for every set of valuations:
\begin{table}[h]\center
\begin{tabular}{c|c|c|c|c}
bidder&value&$\alpha_{i,1}$&$\alpha_{i,2}$&$\alpha_{i,3}$\\
\hline
1&$v_1$&1&1&0\\
2&$v_2$&1&1&0\\
3&$v_3$&1&0.5&0.5\\
4&$v_4$&1&0.5&0.5
\end{tabular}
\vspace{0mm}
\caption{\label{tab:e4x3}An example in which no pure-strategy equilibrium exists for all $v$ with a fixed, bid-independent tie-breaking rule.}
\vspace{-0mm}
\end{table}

The example in Table~\ref{tab:e4x3} uses similar techniques to the simpler one in Section~\ref{sec:eff}, so we will only sketch the reasoning here. It is straightforward to argue that any equilibrium must achieve the efficient allocation, otherwise some bidder could deviate and benefit. In Section 3, we saw that it was important to break ties in favor of the bidder who had a greater incremental value for slot 1 over slot 2. In this example, if the efficient allocation chooses bidders 1 and 2 (as well as either bidder 3 or 4), then we see the same structure replicated here --- it will be important to break ties in favor of bidder 3 and/or 4. On the other hand, if the efficient allocation chooses bidders 3 and 4, with one of bidder 1 or 2, then the same structure arises across slots 1 and 3. However, bidders 1 and 2 have a greater incremental value for slot 1 over slot 3 and therefore it is important to break ties in their favor. Thus, any tie-breaking rule that does not depend on bids will necessarily fail for at least one of these scenarios.

\subsection{The importance of the order of sale}

We just saw that selling slots in a different order can be beneficial, but is it ever necessary? In fact, we show that it is.

\begin{observation}
With four bidders and three slots, there exist values and click-through-rates such that the VCG result can be achieved, but not by selling slots in order.
\end{observation}

\begin{table}[h!]\center
\begin{tabular}{c|c|c|c|c}
bidder&value&$\alpha_{i,1}$&$\alpha_{i,2}$&$\alpha_{i,3}$\\
\hline
1&10&1&0.4&0.4\\
2&8&1&0.75&$\frac17$\\
3&8&1&0.5&0.5\\
4&5&1&1&0
\end{tabular}
\caption{\label{tab3} A four-bidder, three-slot example demonstrating that selling items out of order may facilitate VCG results.}
\end{table}

\begin{proof}
Consider the four-bidder, three-slot example depicted in Table \ref{tab3}.  One
can check that the optimal assignment is (1,2,3) and VCG prices for the slots
are $p=[7,5,1]$. If slots are sold in order, then bidder 4 must set $p_3$.
Thus, bidder 4 must be bidding such that $\alpha_{4,3}b_4=p_3$, which implies
$0\times b_4 = 1$. Clearly, this is not possible, and there will be similar
problems even if we require that $\alpha_{4,3}$ is strictly positive.

However, VCG prices can be achieved by selling slots in the order 1,3,2. One can check that the bids $b=[10,7,7,5]$ achieve VCG prices.

{\em Remark:} Note that bidder 1 is indifferent between slots 1 and 3 at VCG prices while bidder 2 strictly prefers slot 2 to 3. Thus, it might seem more natural to sell slots in the order 3,1,2 and have bidder 1 set the price for slot 3. However, one can check that this will fail because  we cannot sell slot 1 after slot 3. Instead, the example is constructed carefully so that bidder 2 can also set the price of slot 3 despite her strict preference for slot 2 at VCG prices.
\end{proof}

\subsection{An auction with expressive bidding}

Finally, we show how we can build an auction that always yields the VCG result as an equilibrium by selling slots in a different order. For this mechanism, we will need bidders to place a distinct bid $b_{i,j}$ for each slot (WLOG we ignore $\alpha$ values here). First, we need to argue that an appropriate ordering exists, then we will construct a mechanism that exploits this ordering.

\subsubsection{Price support orderings and forests}

We first establish that the VCG result is a feasible outcome of an iterated
auction with expressive bidding. If $i$ is paying price $p_i$, then some other
bidder who has not already been allocated is bidding $p_i$ for the slot $i$
wins. It is not a priori clear that this is possible without requiring some
bidder to overbid her true value. We call an ordering that achieves this a {\em
price support ordering} (PSO).

Our first lemma shows that a price support ordering always exists for VCG prices More specifically, we show that a {\em price support forest} (PSF) exists --- a price support forest is a directed forest that captures the ability of bidders to support prices:
\begin{definition}
A {\em price support forest} (PSF) for prices $p_j$ with $n$ slots and bidders is a graph $F$ on $n$ nodes with the following properties:
\begin{itemize}[topsep=4pt,itemsep=-0.0ex,partopsep=1ex,parsep=1ex]
\item $F$ is a directed forest with edges pointing away from the roots.
\item Root nodes (nodes with no incoming edges) have price $p_j=0$.
\item Edge $(i,j)$ in $F$ implies that bidder $i$ can set the price for slot $j$ without overbidding.
\end{itemize}
\end{definition}
\vspace{1mm}

We will formalize ``$i$ can set the price for slot $j$'' below.

Assume that the VCG mechanism assigns bidder $i$ to slot $i$, and let $p_j$ denote the minimum Walrasian equilibrium price for slot $j$ (the VCG price of bidder $j$). The following lemma says that VCG prices always admit a PSF in which edges capture indifferences. A precisely equivalent lemma appears in~\cite{mehta13}, but we include our own version for completeness.

\begin{lemma}[VCG Price Support Lemma]\label{lem:psf}
There exists a directed forest $F$ with the following property: for any slot $j$,  either $p_j=0$, or there is an edge $(i,j)$ corresponding to a bidder who is indifferent between getting slot $i$ at price $p_{i}$ and getting slot $j$ at price $p_j$, ergo $i$ is happy to bid $b_{i,j}=p_j$ for slot $j$ and thereby set its price. Thus, $F$ is a price support forest.
\end{lemma}

\begin{proof} We show how to construct a PSF. WLOG, suppose there are $n$ advertisers and $n$ slots. Construct a directed graph $G$ with $n$ nodes in which there is an edge from node $i$ to node $j$ if advertiser $i$ is indifferent between getting node $i$ at price $p_i$ and getting node $j$ at price $p_j$, that is, $v_{i,i}-p_i=v_{i,j}-p_j$. Note that envy-freeness implies $v_{i,i}-p_i\geq v_{i,j}-p_j$ for all $(i,j)$, so the absence of an edge in $G$ means $v_{i,i}-p_i> v_{i,j}-p_j$.

{\em Claim: Every node $i$ in the graph is reachable from some node $j$ with price $p_j=0$.} Proof by contradiction. If not, then let $S\subseteq[n]$ be the set of nodes that are not reachable from a node with price zero. Let $\delta>0$ be a constant sufficiently small that it has the following properties:
\begin{itemize}
\item $v_{i,i}-p_i\geq\delta+ v_{i,j}-p_j$ for any $(i,j)$ where $v_{i,i}-p_i>v_{i,j}-p_j$ (note that this is any $(i,j)$ that is {\em not} an edge in the graph), and
\item $\delta\leq\min_{j\in S}p_j$.
\end{itemize}Now, consider prices $p'$ that uniformly lower prices for slots in $S$ by $\delta$, keeping other prices fixed:
\[p_j'=\begin{cases}p_j-\delta,&j\in S\\ p_j&otherwise.\end{cases}\]
By construction, we still have $v_{i,i}-p_i'\geq v_{i,j}-p_j'$ for every $(i,j)$, hence these prices $p'$ are envy-free. Since $p_j'\leq p_j$ for all $j$, envy-freeness of $p'$ contradicts the fact that VCG prices are the minimum envy-free prices, proving the claim.

From $G$, compute a PSF $F$ by computing a spanning forest of $G$.
\end{proof}

\begin{corollary}\label{cor:pso}
There exists an ordering $\sigma$ of slots with the following property: for any slot $j$,  either $p_j=0$, or there is some bidder with $i>j$ who is indifferent between getting slot $i$ at price $p_{i}$ and getting slot $j$ at price $p_j$, ergo $i$ is happy to bid $b_{i,j}=p_j$ for slot $j$ and thereby set its price.
\end{corollary}

\begin{proof}
By Lemma~\ref{lem:psf}, we know that a price support forest $F$ exists. Compute an ordering $\sigma$ such that any parent in $F$ comes after all its children (e.g.) by a breadth-first traversal of $F$.
\end{proof}

\if 0
\begin{observation} A price-support ordering may sell ``less preferred'' slots first. For example, even when slots are ordered, a price-support ordering may not sell slots according to that order.
\end{observation}

\begin{proof} Consider the following example with 3 bidders (A, B, and C) and 3 slots (1, 2, and 3):

\begin{center}
\begin{tabular}{c|c|c|c}
Bidder&Value&$\alpha_{i,1}$&$\alpha_{i,2}$\\
\hline
1&2&1&1\\
2&2&1&1\\
3&1&1&0\\
\end{tabular}
\end{center}

The optimal assignment gives slots 1 and 2 to bidders A and B, while C gets slot 3 (i.e. nothing). If either A or B is removed, the optimal assignment shows C in the top slot and A/B in the second slot. As a result, the VCG prices for A and B are both 1. However, since C has no value for slot 2, her bid cannot set the price for slot 2 without bidding above her value. Consequently, the only ordering that could possibly support VCG prices without overbidding is (2,1,3). This is the ordering given by Lemma~\ref{lem:pso}.

\end{proof}
\fi

\subsubsection{Auctions leveraging price support}

Finally, we show how the existence of a PSF can be used to construct an auction
that supports the VCG result as an equilibrium:

\begin{definition}[Auction with a Price Support Order]
An iterated second-price auction can be implemented leveraging a price support order as follows:

\begin{enumerate}[topsep=2pt,itemsep=-0.0ex,partopsep=1ex,parsep=1ex]
\item Choose an order of sale $\sigma$ and tie-breaking rules that maximize seller revenue given bids. If a slot has only one nonzero bid, it does not get sold.
\item Run an iterated second-price auction according to the order $\sigma$ and rules selected in (1).
\end{enumerate}
\end{definition}

\begin{theorem}[Equilibrium] The iterated second-price auction with unit-demand bidders and expressive bids has an efficient equilibrium in which bidders pay VCG prices.

\end{theorem}

\begin{proof} Choose an arbitrary PSF and define bids as follows:
\[b_{i,j} = \begin{cases}v_{i,j}&i=j\\ p_j&\mbox{$i$ supports the price of $j$ in the PSF, or}\\0 &\mbox{otherwise.}\end{cases}\]

Consider deviations by a particular bidder $k$. Notice that no slot $j$ can have a bid less than $p_j$ unless it was placed by bidder $k$. We can thus conclude that if $k$ wins any slot for less than the VCG price through this defection (which is necessary for it to be profitable), then $k$ must have won the slot for free. However, our auction rules stipulate that slots will not be sold if they only have one nonzero bid, so this is impossible. 
\end{proof}

%==================================================================

\section{Conclusion} \label{sec:conc}

The primary theoretical justification for GSP builds on the analyses of
\cite{varian07} and \cite{eos07} to argue that GSP will perform at least as
well as VCG.  Unfortunately, our results demonstrate that this is a very
fragile phenomenon---when GSP is naturally generalized as an iterated second
price auction, these performance guarantees fall apart even with small
deviations from GSP's separable model. Our work suggests a few techniques for
recovering desirable performance guarantees, such as varying the order of sale
and allowing expressive bidding, but perhaps even more importantly it points to
significant open questions that might suggest new mechanisms and principles for
implementing auctions:

\begin{itemize}[topsep=4pt,itemsep=-0.0ex,partopsep=1ex,parsep=1ex]
\item Is there a better way to generalize GSP that would preserve the performance guarantees of \cite{varian07} and \cite{eos07}?
\item What are the key principles that define GSP in theory?
\item What are the properties that capture GSP's practical popularity?
\end{itemize}

That said, our results also include a surprising positive result: {\em all}
click-through-rate profiles ensure existence of efficient equilibria in the
two-slot setting, given a specific bid-independent tie-breaking rule. We proved
that this result does not generalize to the case with more slots, but whether
bid-{\em dependent} tie-breaking rules could yield generic existence of
efficient equilibria remains an open question.
And even if no meaningful extension beyond the two-slot setting is possible,
the positive result we have may turn out to be relevant in a world of mobile
devices where only a small number of slots can be shown per page.

%{\small 
\bibliographystyle{named}
\bibliography{cw-gsp1}
%}

\pagebreak

\section*{Appendix}

%-------------------------------------------------------
\subsection*{Proofs for Section \ref{sec:eff}}

The proof of Theorem \ref{the:eff} makes use of two important lemmas, the first
giving sufficient conditions for existence of an efficient equilibrium, and the
second demonstrating a relationship between event-ratios and efficient
allocations among pairs of agents.

\begin{lemma}   \label{lem:eff}
Consider a two-slot setting with $n \geq 3$ bidders, arbitrary weakly
decreasing\footnote{I.e., $\a_{i,1} \geq \a_{i,2}$, $\forall i \in I$.} and
click-through-rates, arbitrary values, and an arbitrary efficient allocation, letting 1
and 2 denote the respective winners of slots 1 and 2. Bids $b$ yield this
allocation in equilibrium if (A0)--(A6) hold or (B1)--(B6) and the following
condition on tie-breaking holds: if $\a_{2,1}b_2=\max_{j \in I\setminus\{1,2\}}
\a_{j,1} b_j$, if 1 were hypothetically to underbid this value, slot 1 would be
allocated to $\argmax_{j \in I\setminus\{1,2\}} \a_{j,1} b_j$.

\hspace{-12mm}
\begin{minipage}{.5\linewidth}
\begin{align*}
&\a_{2,1} b_2 \geq \max_{j \in I\setminus\{1,2\}} \a_{j,1} b_j  \tag{A0}
\\ &\a_{2,1} b_2 < \a_{1,1} v_1  \tag{A1}
\\ &\a_{2,2} b_2 > \max_{i \in I\setminus\{1,2\}} \a_{i,2} b_i  \tag{A2}
\\ &\max_{i \in I\setminus\{1,2\}} \a_{i,2} b_i \geq \a_{2,1}b_2 - (\a_{1,1}-\a_{1,2})v_1  \tag{A3}
\\ &\a_{2,2} b_2 \geq \max_{k \in I\setminus\{1,2\}} \a_{k,2} v_k  \tag{A4}
\\ &\max_{i \in I\setminus\{1,2\}} \a_{i,2} b_i \leq \a_{1,1}v_1 - (\a_{2,1}-\a_{2,2}) v_2  \tag{A5}
\\ &\max_{i \in I\setminus\{1,2\}} \a_{i,2} b_i \leq \a_{2,2} v_2  \tag{A6}
\end{align*}
\end{minipage} \hspace{5mm}
\begin{minipage}{.5\linewidth}
\begin{align*}
&\max_{j \in I\setminus\{1,2\}} \a_{j,1} b_j \geq \a_{2,1} b_2  \tag{B0}
\\ &\max_{j \in I\setminus\{1,2\}} \a_{j,1} b_j < \a_{1,1} v_1  \tag{B1}
\\ &\a_{2,2} b_2 > \max_{i \in I\setminus\{1,2\}} \a_{i,2} b_i  \tag{B2}
\\ &\max_{j \in I\setminus\{1,2\}} \a_{j,1} b_j \leq \a_{2,2}b_2 + (\a_{1,1}-\a_{1,2})v_1  \tag{B3}
\\ &\a_{2,2} b_2 \geq \max_{k \in I\setminus\{1,2\}} \a_{k,2} v_k  \tag{B4}
\\ &\max_{i \in I\setminus\{1,2\}} \a_{i,2} b_i \leq \a_{1,1}v_1 - (\a_{2,1}-\a_{2,2}) v_2  \tag{B5}
\\ &\max_{i \in I\setminus\{1,2\}} \a_{i,2} b_i \leq \a_{2,2} v_2  \tag{B6}
\end{align*}
\end{minipage}

\end{lemma}

\begin{proof}
Take $b_1=v_1$.
The (A0)/(B0) condition dichotomizes the set of possible bids into those where
2 is setting the price for 1 (A) and those where some other agent is.

(A1) and (B1)---in their respective contexts of (A0) and (B0)---imply that 1
wins slot 1 and has no incentive to deviate in a way that gives him no slot
(for price 0).

(A2) and (B2) entail that 2 receives slot 2.

(A3) and (B3) entail that 1 has no incentive to deviate in a way that gives him
slot 2, for price $\a_{3,2}b_3$ in the A case and $\a_{2,2}b_2$ in the B case.
(A3) is sufficient even if a hypothetical tie for slot 1 between 2 and 3 is
broken in favor of 3, due to (A2). In the case of (B3) we are using the
tie-breaking assumption in the lemma statement.

(A4) and (B4) entail that no losing agent has an incentive to deviate in a way
that yields him slot 2. (A5) and (B5) entail that 2 has no incentive to bid in
a way that instead yields him slot 1; (A6) and (B6) entail that 2 has no
incentive to bid in a way that instead yields him no slot (for price 0). We
know by efficiency and the fact that $b_1=v_1$ that no losing agent can benefit
by bidding to receive slot 1 (for price $\a_{1,1}v_1$).
\end{proof}

\begin{lemma}  \label{lem2i}
For arbitrary values $v$ and click-through-rates $\a$, letting 1 and 2 denote
the respective winners of slots 1 and 2 in the efficient allocation, $\forall i
\in I\setminus\{1,2\}$,
%
%\[ (\a_{i,1}-\a_{i,2})v_i > (\a_{2,1}-\a_{2,2})v_2 \;\;\Rightarrow\;\;
%\frac{\a_{2,1}}{\a_{2,2}} < \frac{\a_{i,1}}{\a_{i,2}}  \]
\[ \frac{\a_{2,1}}{\a_{2,2}} \geq \frac{\a_{i,1}}{\a_{i,2}} \Rightarrow
(\a_{2,1}-\a_{2,2})v_2 \geq (\a_{i,1}-\a_{i,2})v_i  \]
\end{lemma}

\begin{proof}
Consider arbitrary $i \in I\setminus\{1,2\}$. We have:
\begin{align*}
\a_{2,2} v_2 > \a_{i,2} v_i
 \;\Rightarrow\; &\frac{\a_{2,1}-\a_{2,2}}{\a_{i,2}} v_2 \geq \frac{\a_{2,1}-\a_{2,2}}{\a_{2,2}} v_i
\\ \Leftrightarrow\; &\frac{\a_{i,2}v_i + (\a_{2,1}-\a_{2,2})v_2}{\a_{i,2}} \geq \frac{\a_{2,1}}{\a_{2,2}} v_i
\\ \Leftrightarrow\; &\frac{\a_{i,2}v_i + (\a_{2,1}-\a_{2,2})v_2}{\a_{2,1}} \geq \frac{\a_{i,2}}{\a_{2,2}} v_i
\end{align*}

\noindent Since the first inequality holds by efficiency, the last inequality holds too.
Now note that:
\begin{align*}
&(\a_{i,1}-\a_{i,2})v_i > (\a_{2,1}-\a_{2,2})v_2
 \;\;\Leftrightarrow\;\; \frac{\a_{i,1}}{\a_{2,1}} v_i > \frac{\a_{i,2}v_i + (\a_{2,1}-\a_{2,2})v_2}{\a_{2,1}}
\end{align*}

\noindent Therefore $(\a_{i,1}-\a_{i,2})v_i > (\a_{2,1}-\a_{2,2})v_2
\Rightarrow \frac{\a_{i,1}}{\a_{2,1}} v_i > \frac{\a_{i,2}}{\a_{2,2}} v_i$,
i.e., $(\a_{i,1}-\a_{i,2})v_i > (\a_{2,1}-\a_{2,2})v_2 \Rightarrow
\frac{\a_{i,1}}{\a_{i,2}} > \frac{\a_{2,1}}{\a_{2,2}}$. We state the lemma in
the form of the contrapositive, for more direct application in the results that
follow.
\end{proof}

\begin{theore}[\ref{the:eff}]
For a two-slot setting with any number of bidders, for arbitrary values and
click-through-rates, if there is a unique efficient allocation and ties are broken in
favor of an agent with highest click-ratio, then GSP has an efficient
equilibrium without overbidding.
\end{theore}

\begin{proof}
If there is one bidder the result holds trivially. If there are two or more
bidders, let 1 and 2 denote the respective winners of slots 1 and 2 in the
efficient allocation. If there are exactly two bidders, it is easy to verify
that $b_1=v_1$ and $b_2=0$ is an equilibrium.
So assume there are $n \geq 3$ bidders, and let 3 denote $\argmax_{j \in
I\setminus\{1,2\}} \a_{j,2}v_j$.

First assume $\frac{\a_{2,1}}{\a_{2,2}} \geq \frac{\a_{3,1}}{\a_{3,2}}$, and
consider the following bid profile: $b_1=v_1$,
$b_2=\frac{\a_{3,2}v_3+(\a_{2,1}-\a_{2,2})v_2}{\a_{2,1}}$, $b_3=v_3$, and
$b_i=0$, $\forall i \in I\setminus\{1,2,3\}$. Note that no agent overbids in
this profile, and all bids are non-negative.\footnote{To see that $b_2 \geq 0$,
note that $\frac{\a_{2,1}}{\a_{2,2}} \geq \frac{\a_{3,1}}{\a_{3,2}}$ implies
$(\a_{3,1}-\a_{3,2}) v_3 \leq (\a_{2,1}-\a_{2,2})v_2$, by Lemma \ref{lem2i}.
This in turn implies $\a_{2,2}v_2 \leq \a_{2,1}v_2 + \a_{3,2}v_3$, since
click-through-rates and values are non-negative.} We will show that (A0)--(A6) hold.

(A0) reduces to $(\a_{3,1}-\a_{3,2}) v_3 \leq (\a_{2,1}-\a_{2,2})v_2$, which
holds by Lemma \ref{lem2i}.

(A1) reduces to $\a_{2,1}v_2+\a_{3,2}v_3 < \a_{1,1}v_1+\a_{2,2}v_2$, which
holds by efficiency.

(A2) reduces to $\a_{3,2}v_3 < \a_{2,2}v_2$, which holds by efficiency.

(A3) reduces to $(\a_{2,1}-\a_{2,2})v_2 \leq (\a_{1,1}-\a_{1,2})v_1$, which
holds by efficiency.

(A4) reduces to $\a_{3,2}v_3 \leq \a_{2,2}v_2$, which, like (A2), holds by
efficiency.

(A5) reduces to $\a_{2,1}v_2+\a_{3,2}v_3 \leq \a_{1,1}v_1+\a_{2,2}v2$, which holds by efficiency.

(A6) reduces to $\a_{3,2}v_3 \leq \a_{2,2}v_2$, which, like (A2) and (A4),
holds by efficiency.

Now, to prove the theorem it is sufficient to show that if
$\frac{\a_{2,1}}{\a_{2,2}} < \frac{\a_{3,1}}{\a_{3,2}}$, there exists a set of
bids satisfying (B0)--(B6).
Assume $\frac{\a_{2,1}}{\a_{2,2}} < \frac{\a_{3,1}}{\a_{3,2}}$ and take
$b_1=v_1$, $b_2=\frac{\a_{3,2}}{\a_{2,2}} v_3$,
$b_3=\frac{\a_{2,1}}{\a_{2,2}}\frac{\a_{3,2}}{\a_{3,1}}v_3$, and $b_i=0$,
$\forall i \in I\setminus\{1,2,3\}$. Again, no agent overbids in this
profile---in fact, 2 is bidding the minimum he would need to bid if 3 were to
bid truthfully, and 3 is underbidding (since
$\frac{\a_{2,1}}{\a_{2,2}}\frac{\a_{3,2}}{\a_{3,1}}<1$). Moreover, if ties are
broken by click-ratio, then any hypothetical ties between 2 and 3 will be
broken in favor of 3, which satisfies the conditions of Lemma \ref{lem:eff}.

(B0) holds with equality.

(B1) follows from efficiency, since 3 is not overbidding.

(B2) reduces to $1 > \frac{\a_{2,1}}{\a_{2,2}}\frac{\a_{3,2}}{\a_{3,1}}$, which
holds if $\frac{\a_{2,1}}{\a_{2,2}} < \frac{\a_{3,1}}{\a_{3,2}}$.

(B3) reduces to $\a_{3,2}(\a_{2,1}-\a_{2,2})v_3 \leq \a_{2,2}
(\a_{1,1}-\a_{1,2})v_1$. To see that this holds, note that efficiency entails
that $\a_{3,2} v_3 \leq \a_{2,2}v_2$ and $(\a_{2,1}-\a_{2,2})v_2 \leq
(\a_{1,1}-\a_{1,2})v_1$.

(B4) holds with equality.

(B5) reduces to $\a_{2,1}v_2 +
\frac{\a_{2,1}}{\a_{2,2}}\frac{\a_{3,2}}{\a_{3,1}}\a_{3,2}v_3 \leq
\a_{1,1}v_1+\a_{2,2}v_2$, which follows from efficiency and the fact that
$\frac{\a_{2,1}}{\a_{2,2}}\frac{\a_{3,2}}{\a_{3,1}} < 1$.

(B6) follows from (B2) and the fact that $b_2 \leq v_2$.
\end{proof}

%-------------------------------------------------------

%\vspace{166mm} \pagebreak
\subsection*{Proofs for Section \ref{sec:gef}}

The missing proof from Section \ref{sec:gef} is that of Theorem \ref{the:gef}.
On our way to proving that, we establish the following two lemmas,
characterizing global envy-freeness for the two-item, $n$-bidder case (Lemma
\ref{lem:gef0}), and then consolidating those constraints with efficiency
constraints for the three-bidder case (Lemma \ref{lemE2}). Lemma \ref{lem:gef0}
will be important for the proof of Theorem \ref{the:gef2}, omitted from the main
text and included below.

\begin{lemma}  \label{lem:gef0}
For arbitrary click-through-rates $\a$, values $v$, and bids $b$, the resulting outcome
is globally envy-free if and only if, letting 1 denote $\argmax_{i\in I}
\a_{i,1}b_i$, 2 denote $\argmax_{i\in I\setminus\{1\}} \a_{i,2}b_i$, 3 denote
$\argmax_{i\in I\setminus\{1,2\}} \a_{i,2}b_i$, and 4 denote $\argmax_{i\in
I\setminus\{1,2\}} \a_{i,1}b_i$: (A-GEF0)--(A-GEF6) $\vee$ (B-GEF0)--(B-GEF6)
$\vee $ (C-GEF0)--(C-GEF6).

\hspace{-15mm}
\begin{minipage}{.5\linewidth}
\begin{align*}
&b_2 \geq \frac{\a_{4,1}}{\a_{2,1}}b_4  \tag{A-GEF0}
\\&b_2 \leq \frac{\a_{3,2}b_3 + (\a_{1,1}-\a_{1,2})v_1}{\a_{2,1}}  \tag{A-GEF1}
\\ &b_2 \leq \frac{\a_{1,1}}{\a_{2,1}}v_1  \tag{A-GEF2}
\\ &b_2 \geq \frac{\a_{3,2}b_3 + (\a_{2,1}-\a_{2,2})v_2}{\a_{2,1}}  \tag{A-GEF3}
\\ &b_3 \leq \frac{\a_{2,2}}{\a_{3,2}}v_2  \tag{A-GEF4}
\\ &b_2 \geq \frac{1}{\a_{2,1}}\max_{i \in I\setminus\{1,2\}} \a_{i,1} v_i  \tag{A-GEF5}
\\ &b_3 \geq \frac{1}{\a_{3,2}}\max_{i \in I\setminus\{1,2\}} \a_{i,2} v_i  \tag{A-GEF6}
\end{align*}
\end{minipage} \hspace{5mm}
\begin{minipage}{.5\linewidth}
\begin{align*}
&b_2 \leq \frac{\a_{3,1}}{\a_{2,1}}b_3,\; b_3 = \frac{\a_{4,1}}{\a_{3,1}}b_4  \tag{B-GEF0}
\\ &b_3 \leq \frac{\a_{1,1}-\a_{1,2}}{\a_{3,1}-\a_{3,2}}v_1  \tag{B-GEF1}
\\ &b_3 \leq \frac{\a_{1,1}}{\a_{3,1}}v_1  \tag{B-GEF2}
\\ &b_3 \geq \frac{\a_{2,1}-\a_{2,2}}{\a_{3,1}-\a_{3,2}}v_2  \tag{B-GEF3}
\\ &b_3 \leq \frac{\a_{2,2}}{\a_{3,2}}v_2  \tag{B-GEF4}
\\ &b_3 \geq \frac{1}{\a_{3,1}}\max_{i \in I\setminus\{1,2\}} \a_{i,1} v_i  \tag{B-GEF5}
\\ &b_3 \geq \frac{1}{\a_{3,2}}\max_{i \in I\setminus\{1,2\}} \a_{i,2} v_i  \tag{B-GEF6}
\end{align*}
\end{minipage}

\hspace{-25mm} \begin{minipage}{1.0\linewidth} \begin{align*}
&b_2 \leq \frac{\a_{4,1}}{\a_{2,1}}b_4,\; b_3 < \frac{\a_{4,1}}{\a_{3,1}}b_4  \tag{C-GEF0}
\\ &b_4 \leq \frac{\a_{3,2}b_3+(\a_{1,1}-\a_{1,2})v_1}{\a_{4,1}}  \tag{C-GEF1}
\\ &b_4 \leq \frac{\a_{1,1}}{\a_{4,1}}v_1  \tag{C-GEF2}
\\ &b_4 \geq \frac{\a_{3,2}b_3+(\a_{2,1}-\a_{2,2})v_2}{\a_{4,1}}  \tag{C-GEF3}
\\ &b_3 \leq \frac{\a_{2,2}}{\a_{3,2}}v_2  \tag{C-GEF4}
\\ &b_4 \geq \frac{1}{\a_{4,1}}\max_{i \in I\setminus\{1,2\}} \a_{i,1} v_i  \tag{C-GEF5}
\\ &b_3 \geq \frac{1}{\a_{3,2}}\max_{i \in I\setminus\{1,2\}} \a_{i,2} v_i  \tag{C-GEF6}
\end{align*} \end{minipage}
\end{lemma}

\begin{proof}
Consider arbitrary click-through-rates $\a$, values $v$, and bids $b$ without
overbidding. Let $p_1$ and $p_2$ denote the prices paid by the winners of slots
1 and 2, respectively. Since no agent overbids, $p_2 \leq \a_{3,2}v_3$. If
$b_3< v_3$ then $p_2 < \a_{3,2}v_3$ and agent 3 is envious. Therefore, in any
globally envy-free outcome, $b_3=v_3$; or, to be precise, if there is a tie for
definition of agent 3, one such agent must bid his value, and we will call him
agent 3 in what follows.

Now there are three possibilities: (i) 2 sets $p_1$ and 3 sets $p_2$, (ii) 3
sets both $p_1$ and $p_2$, or (iii) 3 sets $p_2$ and some other agent 4 sets
$p_1$. (A-GEF0)--(A-GEF6) encode exactly the envy constraints for possibility
(i): 2 sets $p_1$ (A-GEF0), 1 doesn't envy 2 (A-GEF1), 1 doesn't envy
an unallocated agent (A-GEF2), 2 doesn't envy 1 (A-GEF3), 2 doesn't envy an
unallocated agent (A-GEF4), no unallocated agent envies 1 (A-GEF5), and no
unallocated agent envies 2 (A-GEF6).

(B-GEF0)--(B-GEF6) encode exactly the envy constraints for possibility (ii): 3
sets $p_1$ and $p_2$ (B-GEF0), 1 doesn't envy 2 (B-GEF1), 1 doesn't envy an
unallocated agent (B-GEF2), 2 doesn't envy 1 (B-GEF3), 2 doesn't envy an
unallocated agent (B-GEF4), no unallocated agent envies 1 (B-GEF5), and no
unallocated agent envies 2 (B-GEF6).

Finally, (C-GEF0)--(C-GEF6) encode exactly the envy constraints for possibility
(iii): 4 sets $p_1$ (C-GEF0), 1 doesn't envy 2 (C-GEF1), 1 doesn't envy an
unallocated agent (C-GEF2), 2 doesn't envy 1 (C-GEF3), 2 doesn't envy an
unallocated agent (C-GEF4), and no unallocated agent envies 1 (C-GEF5), and no
unallocated agent envies 2 (C-GEF6).
\end{proof}

\begin{lemma}  \label{lemE2}
In a two-slot, three-bidder setting, for arbitrary click-through-rates $\a$ and values
$v$, letting 1 and 2 denote the respective winners of slots 1 and 2 in the
efficient allocation and 3 the other bidder, if ties are broken in favor of an
agent with highest click-ratio: arbitrary bids $b$ with $\a_{1,1}b_1 \geq
\max\{\a_{2,1}b_2,\,\a_{3,1}b_3\}$ constitute an efficient and globally
envy-free equilibrium if (D0)--(D7) or (E0)--(E8).

\hspace{-18mm}
\begin{minipage}{.5\linewidth}
\begin{align*}
&b_2 \geq \frac{\a_{3,1}}{\a_{2,1}}b_3  \tag{D0}
\\&b_2 \leq \frac{\a_{3,2}b_3 + (\a_{1,1}-\a_{1,2})v_1}{\a_{2,1}}  \tag{D1}
\\ &b_2 < \frac{\a_{1,1}}{\a_{2,1}}v_1  \tag{D2}
\\ &b_2 \geq \frac{\a_{3,2}b_3 + (\a_{2,1}-\a_{2,2})v_2}{\a_{2,1}}  \tag{D3}
\\ &b_3 \leq \frac{\a_{2,2}}{\a_{3,2}}v_2  \tag{D4}
\\ &b_3 \geq v_3  \tag{D5}
%%%%%
\\ &b_2 > \frac{\a_{3,2}}{\a_{2,2}} b_3  \tag{D6}
\\ &b_3 \leq \frac{\a_{1,1}v_1 - (\a_{2,1}-\a_{2,2}) v_2}{\a_{3,2}}  \tag{D7}
\end{align*}
\end{minipage} \hspace{3mm}
\begin{minipage}{.5\linewidth}
\begin{align*}
&b_2 \leq \frac{\a_{3,1}}{\a_{2,1}}b_3  \tag{E0}
\\ &b_3 \leq \frac{\a_{1,1}-\a_{1,2}}{\a_{3,1}-\a_{3,2}}v_1  \tag{E1}
\\ &b_3 \geq \frac{\a_{2,1}-\a_{2,2}}{\a_{3,1}-\a_{3,2}}v_2  \tag{E2}
\\ &b_3 \leq \frac{\a_{2,2}}{\a_{3,2}}v_2  \tag{E3}
\\ &b_3 \geq v_3  \tag{E4}
%%%%%
\\ &b_3 < \frac{\a_{1,1}}{\a_{3,1}} v_1  \tag{E5}
\\ &b_2 > \frac{\a_{3,2}}{\a_{2,2}} b_3  \tag{E6}
\\ &b_3 \leq \frac{\a_{2,2}b_2 + (\a_{1,1}-\a_{1,2})v_1}{\a_{3,1}}  \tag{E7}
\\ &b_3 \leq \frac{\a_{1,1}v_1 - (\a_{2,1}-\a_{2,2}) v_2}{\a_{3,2}}  \tag{E8}
\end{align*}
\end{minipage}

\noindent $b$ does not constitute an efficient and globally envy-free
equilibrium unless (D0)--(D7) or (E0)--(E8) hold, replacing the strict
inequalities with weak inequalities.
\end{lemma}

\begin{proof}
The lemma follows directly from combining the efficiency equilibrium constraints
and the envy-freeness constraints, found in Lemmas \ref{lem:eff} and
\ref{lem:gef0}, reducing the constraints as the 3-agent case allows.
\end{proof}

\if 0
\begin{propositio}[\ref{prop:gef0}]
%
In a two-slot, three-bidder setting, for arbitrary values $v$ and click-through-rates
$\a$, there exist no bids yielding a globally envy-free outcome unless, letting
1 and 2 denote the respective winners of slots 1 and 2 in the efficient
allocation:
%
\begin{align*} &(\a_{3,1}-\a_{3,2})v_3 \leq (\a_{1,1}-\a_{1,2}) v_1
\end{align*}
%
\end{propositio}

\begin{proof}
%
Take arbitrary values $v$, click-through-rates $\a$, and bids $b$.
%
First assume $b_2 \geq \frac{\a_{3,1}}{\a_{2,1}}b_3$ (i.e., 2 sets the price
for 1). For 1 to not be envious of 2, it must be the case that $b_2 \leq
\frac{\a_{3,2}b_3 + (\a_{1,1}-\a_{1,2})v_1}{\a_{2,1}}$. The combination of
these two constraints yields $(\a_{3,1}-\a_{3,2})b_3 \leq (\a_{1,1}-\a_{1,2})
v_1$. Now instead assume $b_2 \leq \frac{\a_{3,1}}{\a_{2,1}}b_3$ (i.e, 3 sets
the price for 1). 1 is not envious of 2 if and only if $\a_{1,1}v_1-\a_{3,1}b_3
\geq \a_{1,2}v_1-\a_{3,2}b_3$, i.e., $(\a_{3,1}-\a_{3,2})b_3 \leq
(\a_{1,1}-\a_{1,2}) v_1$, again. Finally, noting that envy-freeness for 3
requires that $b_3 \geq v_3$ (otherwise 3 would envy 2), global envy-freeness
requires:
%
\[ (\a_{3,1}-\a_{3,2})v_3 \leq (\a_{3,1}-\a_{3,2})b_3 \leq (\a_{1,1}-\a_{1,2})
v_1 \]
%
If $(\a_{3,1}-\a_{3,2})v_3 > (\a_{1,1}-\a_{1,2}) v_1$, this cannot be
satisfied.
%
\end{proof}
\fi

\begin{theore}[\ref{the:gef}]
In a two-slot, three-bidder setting, for arbitrary click-through-rates $\a$ and values
$v$, there exist bids---without overbidding---yielding a globally envy-free and
efficient equilibrium if and only if, letting 1 and 2 denote the respective
winners of slots 1 and 2 in the efficient allocation:
\begin{align*} &(\a_{3,1}-\a_{3,2})v_3 \leq (\a_{1,1}-\a_{1,2}) v_1
\end{align*}
If a globally envy-free and efficient equilibrium exists, one exists that
yields the VCG result and does not require overbidding.
\end{theore}

\begin{proof}
Take arbitrary click-through-rates $\a$ and values $v$. The ``only if'' direction holds
by Proposition \ref{prop:gef0}.  \if 0 (D0) and (D1) combined (i.e., 2 setting
the price for 1, with 1 not envious of 2) entail that $(\a_{3,1}-\a_{3,2})b_3
\leq (\a_{1,1}-\a_{1,2}) v_1$. This is also exactly the (E1) constraint (i.e.,
1 not envious of 2 in a context where 3 is setting both prices). Adding the
$b_3 \geq v_3$ constraint (D5, E4), if $(\a_{3,1}-\a_{3,2})v_3 >
(\a_{1,1}-\a_{1,2}) v_1$ we can see that neither the D nor the E sets of
constraints can be satisfied. \fi
%
So assume $(\a_{3,1}-\a_{3,2})v_3 \leq (\a_{1,1}-\a_{1,2})v_1$. First consider
the case where $(\a_{3,1}-\a_{3,2})v_3 \leq (\a_{2,1}-\a_{2,2})v_2$.  Let
$b_1=v_1$, $b_2=\frac{\a_{3,2}v_3+(\a_{2,1}-\a_{2,2})v_2}{\a_{2,1}}$, and
$b_3=v_3$.
(D0) reduces to $(\a_{3,1}-\a_{3,2})v_3 \leq (\a_{2,1}-\a_{2,2})v_2$, which
holds by assumption. (D1) reduces to $(\a_{3,1}-\a_{3,2})v_3 \leq
(\a_{1,1}-\a_{1,2})v_1$, which also holds by assumption. (D2) reduces to
$\a_{2,1}v_2+\a_{3,2}v_3 < \a_{1,1}v_1+\a_{2,2}v_2$, which holds by efficiency.
(D3) holds with equality. (D4) reduces to $\a_{3,2}v_3 \leq \a_{2,2}v_2$, which
holds by efficiency. (D5) holds with equality. (D6) reduces to $\a_{3,2}v_3 >
\a_{2,2}v_2$, which holds by efficiency. (D7) reduces to
$\a_{2,1}v_2+\a_{3,2}v_3 \leq \a_{1,1}v_1+\a_{2,2}v_2$, which holds by
efficiency.
Note that these bids yield the VCG result, since when $(\a_{3,1}-\a_{3,2})v_3
\leq (\a_{2,1}-\a_{2,2})v_2  \leq (\a_{1,1}-\a_{1,2})v_1$, the externality
imposed by 1 is $\a_{3,2}v_3+(\a_{2,1}-\a_{2,2})v_2$ and that imposed by 2 is
$\a_{3,2}v_3$. Moreover,
$b_2=\frac{\a_{3,2}v_3+(\a_{2,1}-\a_{2,2})v_2}{\a_{2,1}} = v_2 -
\frac{\a_{2,2}v_2-\a_{3,2}v_3}{\a_{2,1}} \leq v_2$, and so no overbidding is
required.

Now consider the case where $(\a_{3,1}-\a_{3,2})v_3 > (\a_{2,1}-\a_{2,2})v_2$.
We will have to consider one further conditional. Assume first that
$\frac{\a_{3,1}}{\a_{2,1}}v_3 \leq v_2$. In this case, take $b_1=v_1$,
$b_2=\frac{\a_{3,1}}{\a_{2,1}}v_3$, and $b_3=v_3$.
(D0) holds with equality. (D1) reduces to $\a_{3,1}v_3+\a_{1,2}v_1 \leq
\a_{1,1}v_1+\a_{2,2})v_2$, which holds by efficiency. (D2) reduces to
$\a_{3,1}v_3 < \a_{1,1}v_1$, which holds by efficiency.  (D3) reduces to
$(\a_{3,1}-\a_{3,2})v_3 \geq (\a_{2,1}-\a_{2,2})v_2$, which holds by
assumption.  (D4) reduces to $\a_{3,2}v_3 \leq \a_{2,2}v_2$, which holds by
efficiency. (D5) holds with equality. (D6) reduces to
$\frac{\a_{2,1}}{\a_{2,2}} < \frac{\a_{3,1}}{\a_{3,2}}$, which holds by Lemma
\ref{lem2i} (taking the contrapositive). (D7) reduces to
$\a_{2,1}v_2+\a_{3,2}v_3 \leq \a_{1,1}v_1+\a_{2,2}v_2$, which holds by
efficiency. So this is an efficient and GEF equilibrium.
These bids also yield the VCG result, since the externality imposed by 1 is
$\a_{3,1}v_3$, a price which is set by 3 in this case, and that imposed by 2 is
$\a_{3,2}v_3$. Moreover, since $b_2=\frac{\a_{3,1}}{\a_{2,1}}v_3 \leq v_2$ by
assumption, no overbidding is required.
%%%

Now assume instead that $\frac{\a_{3,1}}{\a_{2,1}}v_3 > v_2$. In this case,
take $b_1=v_1$, $b_2=v_2$, and $b_3=v_3$.
(E0) holds by assumption. (E1) reduces to $(\a_{3,1}-\a_{3,2})v_3 \leq
(\a_{1,1}-\a_{1,2})v_1$, which also holds by assumption.  (E2) reduces to
$(\a_{3,1}-\a_{3,2})v_3 \leq (\a_{2,1}-\a_{2,2})v_2$, which again holds by
assumption. (E3) reduces to $\a_{3,2}v_3 \leq \a_{2,2}v_2$, which holds by
efficiency. (E4) holds with equality. (E5) reduces to $\a_{3,1}v_3 <
\a_{1,1}v_1$, which holds by efficiency. (E6) reduces to $\a_{2,2}v_2 >
\a_{3,2}v_3$, which holds by efficiency. (E7) reduces to
$\a_{3,1}v_3+\a_{1,2}v_1 \leq \a_{1,1}v_1+\a_{2,2}v_2$, which holds by
efficiency. Finally, (E8) reduces to $\a_{2,1}v_2+\a_{3,2}v_3 \leq
\a_{1,1}v_1+\a_{2,2}v_2$, which holds by efficiency. Thus this is an efficient
and GEF equilibrium.
These bids also yield the VCG result (with 2 setting the price for 1 in this
case), and since all agents are bidding their true values, clearly no
overbidding is required.

In each of the three conditional cases considered above, an equilibrium was
established with bids that yield the VCG result and do not exceed true
valuations, and so the theorem is proved.
\end{proof}

\begin{theore}[\ref{the:gef2}]
For arbitrary click-through-rates $\a$ and values $v$, letting 1 and 2 denote the
respective winners of slots 1 and 2 in the efficient allocation, if
$\frac{\a_{2,1}}{\a_{2,2}} \geq \frac{\a_{i,1}}{\a_{i,2}}$, $\forall i \in
I\setminus\{1,2\}$, there exists an efficient and globally envy-free
equilibrium without overbidding.
\end{theore}

\begin{proof}
Assume that $\frac{\a_{i,1}}{\a_{i,2}} \leq \frac{\a_{2,1}}{\a_{2,2}}$,
$\forall i \in I\setminus\{1,2\}$. Let 1 and 2 denote the respective winners of
slots 1 and 2 in the efficient allocation, and let 3 denote $\argmax_{i \in
I\setminus\{1,2\}} \a_{i,2}v_i$. Let $b_1=v_1$,
$b_2=\frac{\a_{3,2}v_3+(\a_{2,1}-\a_{2,2})v_2}{\a_{2,1}}$, $b_3=v_3$, and
$b_i=0$, $\forall i \in I\setminus\{1,2,3\}$. The proof of Theorem
\ref{the:eff} demonstrated that these bids satisfy (A0)--(A6) and yield an
efficient equilibrium, without overbidding. Therefore, to prove this theorem it
is sufficient to show that these bids satisfy (A-GEF0)--(A-GEF6).

Noting that the ``agent 4'' of (A-GEF0) is our agent 3, given the above bids,
(A-GEF0) reduces to $(\a_{3,1}-\a_{3,2})v_3 \leq (\a_{2,1}-\a_{2,2})v_2$, and
this holds by Lemma \ref{lem2i}.
(A-GEF1) reduces to $(\a_{3,1}-\a_{3,2})v_3 \leq (\a_{1,1}-\a_{1,2})v_1$, which
holds by $(\a_{3,1}-\a_{3,2})v_3 \leq (\a_{2,1}-\a_{2,2})v_2$ (which we just
demonstrated) combined with $(\a_{2,1}-\a_{2,2})v_2 \leq
(\a_{1,1}-\a_{1,2})v_1$ (which is entailed by efficiency). (A-GEF2) reduces to
$\a_{2,1}v_2+\a_{3,2}v_3 < \a_{1,1}v_1+\a_{2,2}v_2$, which holds by efficiency.
(A-GEF3) holds with equality. (A-GEF4) reduces to $\a_{3,2}v_3 \leq
\a_{2,2}v_2$, which holds by efficiency. (A-GEF5) reduces to
$(\a_{2,1}-\a_{2,2})v_2 \geq \max_{i \in I\setminus\{1,2\}}
\a_{i,2}v_i-\a_{3,2}v_3$. Note that, $\forall i \in I\setminus\{1,2\}$, Lemma
\ref{lem2i} entails that $(\a_{2,1}-\a_{2,2}) v_2 \geq (\a_{i,1}-\a_{i,2})
v_i$. Thus, also using the definition of agent 3, we have:
\[  (\a_{2,1}-\a_{2,2})v_2 \geq (\a_{i,1}-\a_{i,2}) v_i \geq \a_{i,1}v_i -
\a_{3,2}v_3,  \]
and so (A-GEF5) holds. Finally, (A-GEF6) holds with equality by the definition
of agent 3.
\end{proof}

\end{document}